\documentclass[11pt]{article}
\usepackage{graphicx}
\usepackage[utf8x]{inputenc}
\usepackage{amsmath, amssymb, amsthm}
\usepackage{csquotes}
\usepackage[english]{babel}
\usepackage{eurosym}
\usepackage{fullpage}
\usepackage{palatino}
 \usepackage{verbatim}
 \usepackage{fancybox}

\usepackage{subfig}
\usepackage{enumerate}
\usepackage{epsfig}
\usepackage[active]{srcltx}
\usepackage{amsfonts}
\usepackage{amsmath}
\usepackage[mathlines]{lineno}
\usepackage{xspace}
\usepackage{tikz}
\usepackage{authblk}
\usepackage{soul}

\usepackage{algorithmic}
\usepackage{algorithm}
\floatname{algorithm}{Procedure}

\usepackage[textsize=footnotesize,color=green!40]{todonotes}
\definecolor{red}{RGB}{255,0,0}
\definecolor{blue}{RGB}{0,0,255}

\usetikzlibrary{arrows,shapes,snakes,automata,backgrounds,petri,patterns}

\newtheorem{theorem}{Theorem}
\newtheorem{lemma}[theorem]{Lemma}

\newtheorem{claim}{Claim}
\newtheorem{corollary}[theorem]{Corollary}

\newtheorem{question}{Question}

\newcommand{\ds}{\textsc{Dominating Set in Non-Blocking Graphs}}
\newcommand{\dsc}{\textsc{Dominating Set}}
\newcommand{\tss}{\textsc{Token Sliding in Split Graphs}}

\newcounter{claimb}
    
\def\claimb{$$\vcenter\bgroup\advance\hsize by -8em\noindent
\refstepcounter{claimb}\ignorespaces\it}        
\makeatletter
\def\endclaimb{\rm\egroup\leqno(\theclaimb)$$\global\@ignoretrue}
\makeatother

    {\noindent \emph{Proof.} {}{#1}{}}{\hfill
    $\Diamond$\vspace{1em}}

\title{Token Sliding on Chordal Graphs}

\author[1]{Marthe Bonamy}
\author[2]{Nicolas Bousquet}
\affil[1]{CNRS, LaBRI, Universit\'e de Bordeaux, France  \thanks{marthe.bonamy@labri.fr}}
\affil[2]{LIRIS, \'Ecole Centrale de Lyon, France \thanks{nicolas.bousquet@ec-lyon.fr}}

\begin{document}

\maketitle

\begin{abstract}
Let $I$ be an independent set of a graph $G$. Imagine that a token is located on any vertex of $I$. We can now move the tokens of $I$ along the edges of the graph as long as the set of tokens still defines an independent set of $G$. Given two independent sets $I$ and $J$, the \textsc{Token Sliding} problem consists in deciding whether there exists a sequence of independent sets which transforms $I$ into $J$ so that every pair of consecutive independent sets of the sequence can be obtained via a token move. This problem is known to be PSPACE-complete even on planar graphs.

In~\cite{DemaineDFHIOOUY14}, Demaine et al. asked whether the Token Sliding reconfiguration problem is polynomial time solvable on interval graphs and more generally in chordal graphs. Yamada and Uehara~\cite{YamadaU16} showed that a polynomial time transformation can be found in proper interval graphs.

In this paper, we answer the first question of Demaine et al. and generalize the result of Yamada and Uehara by showing that we can decide in polynomial time whether two independent sets of an interval graph are in the same connected component.
Moveover, we answer similar questions by showing that: (i) determining if there exists a token sliding transformation between every pair of $k$-independent sets in an interval graph can be decided in polynomial time; (ii) deciding this problem becomes co-NP-hard and even co-W[2]-hard (parameterized by the size of the independent set) on split graphs, a sub-class of chordal graphs.
\end{abstract}

\section{Introduction}
Reconfiguration problems consist in finding step-by-step transformations between two feasible solutions such that all intermediate results are also feasible. Reconfiguration problems model dynamic situations where a given solution is in place and has to be modified, but no property disruption can be afforded. Two types of questions naturally arise when we deal with reconfiguration problems: (i) when can we ensure that there exist such a transformation? (ii) What is the complexity of finding such a reconfiguration? In the last few years reconfiguration problems received a lot attention for various different problems such as proper colorings~\cite{BonamyB13,Feghali0P15}, Kempe chains~\cite{BonamyBFJ15,FeghaliJP15}, satisfiability~\cite{Gopalan09} or shortest paths~\cite{Bonsma13}. For a complete survey on reconfiguration problems, the reader is referred to~\cite{vHeuvel13}. In this paper our reference problem is independent set.

In the whole paper, $G=(V,E)$ is a graph where $n$ denotes the size of $V$ and $k$ is an integer. For standard definitions and notations on graphs, we refer the reader to~\cite{Diestel2005}.
A \emph{$k$-independent set} of $G$ is a subset $S \subseteq V$ of size $k$ of pairwise non-incident vertices.
The $k$-independent set reconfiguration graph is a graph where vertices are $k$-independent sets and two independent sets are incident if they are ``close'' to each other.
In the last few years, three possible definitions of adjacency between independent sets have been introduced.
In the \emph{Token Addition Removal} (TAR) model~\cite{BonamyB14a,KaminskiMM12,MouawadNRSS13}, two $k$-independent sets $I, J$ are adjacent if they differ on exactly one vertex (\emph{i.e.} if there exists a vertex $u$ such that $I = J \cup \{u\}$ or the other way round). In the \emph{Token Jumping} (TJ) model~\cite{BonsmaKW14,Ito2011,KaminskiMM12}, two independent sets are adjacent if one can be obtained from the other by replacing a vertex with another one (in particular it means that we only look at independent sets of a given size). In the \emph{Token Sliding} (TS) model, first introduced in~\cite{HearnD05}, tokens can be moved along edges of the graph, \emph{i.e} vertices can only be replaced with vertices which are adjacent to them (see~\cite{b14} for a general overview of the results for all these models).

In this paper we concentrate on the Token Sliding (TS) model. Given a graph $G$, the \emph{$k$-TS reconfiguration graph of $G$}, denoted $TS_k(G)$, is the graph whose vertices are $k$-independent sets of $G$ and where two independent sets are incident if we can transform one into the other by sliding a token along an edge. More formally, $I$ and $J$ are adjacent in $TS_k(G)$ if $J \setminus I = \{ u \}$, $I \setminus J = \{ v\}$ and $(u,v)$ is an edge of $G$. We then say that the token on $u$ is slided on $v$.
Hearn and Demaine proved in~\cite{HearnD05} that deciding if two independent sets are in the same connected component of $TS_k(G)$ is PSPACE-complete, even for planar graphs.
On the positive side, Kaminski et al. gave a linear-time algorithm to decide this problem for cographs (which are characterized as $P_4$-free graphs)~\cite{KaminskiMM12}. Bonsma et al.~\cite{BonsmaKW14} showed that we can decide in polynomial time if two independent sets are in the same connected component for claw-free graphs. Demaine et al.~\cite{DemaineDFHIOOUY14} described a quadratic algorithm deciding if two independent sets lie in the same connected component for trees. Yamada and Uehara showed in~\cite{YamadaU16} that a polynomial transformation exists in proper interval graphs.
From a parameterized point of view, Ito et al.~\cite{ItoKOSUY14} showed that the problem is $W[2]$-hard parameterized by the number of tokens. However, for planar graphs, the problem becomes FPT~\cite{ItoKO14}.

\paragraph*{Our contribution.}
In their paper, Demaine et al.~\cite{DemaineDFHIOOUY14} asked if determining whether two independent sets are in the same connected component of $TS_k(G)$ can be decided in polynomial time for interval graphs and then more generally for chordal graphs.
An interval graph is a graph which can be represented as an intersection graph of intervals in the real line. Chordal graphs, that are graphs without induced cycle of length at least $4$, strictly contain interval graphs.
In this paper, we prove the first conjecture and we give a new perspective on these questions by proving that (i) the connectivity of $TS_k(G)$ can be decided in polynomial time if $G$ is an interval graph (ii) deciding if $TS_k(G)$ if connected for chordal graphs is co-NP-hard and co-$W[2]$-hard parameterized by the size $k$ of the independent set even for split graphs, a subclass of chordal graphs.
More formally, the paper is devoted to proving the three following results:

\begin{theorem}\label{th:hard}
The following problem is co-NP-hard and co-$W[2]$-hard parameterized by the size of the independent set.
\begin{quote}
 \tss{} \\
 \textbf{Input:} A split graph $G$, an integer $k$.\\
 \textbf{Output:} YES if and only if $TS_k(G)$ is connected.
\end{quote}
\end{theorem}
Section~\ref{sec:split} is devoted to this result. Our proof consists in exhibiting a reduction from a variant of the \textsc{Dominating Set} problem. A problem is FPT parameterized by $k$ if it can be decided in time $f(k) \cdot n^c$ where $c$ is a constant and $n$ is the size of the instance. An FPT algorithm is deterministic, thus we have FPT=co-FPT. Moreover, the class $W[2]$ is conjectured to strictly contain the class. In particular it means that this problem is unlikely to be solved in FPT-time parameterized by the size of the solution. For more information on parameterized complexity the reader is referred to~\cite{FlumG06}.

On the positive side, we show that the connectivity of $TS_k(G)$ can be decided in polynomial in interval graphs.

\begin{theorem}\label{th:connec}
 Given an interval graph $G$ and an integer $k$, the connectivity of $TS_k(G)$ can be decided in polynomial time.
\end{theorem}

The proof of Theorem~\ref{th:connec} can be adapted in order to answer a question raised by Demaine et al.~\cite{DemaineDFHIOOUY14} and Yamada and Uehara~\cite{YamadaU16}:

\begin{theorem}\label{th:pairsofinterval}
 Given an interval graph $G$ and two independent sets $I$ and $J$ of size $k$, one can decide in polynomial time if $I$ and $J$ are in the same connected component of $TS_k(G)$.
\end{theorem}

In light of the two first results, we ask the following question. The \emph{clique-tree degree} of a chordal graph $G$ is the smallest maximum degree of a clique-tree of $G$.

\begin{question}\label{conj}
For any integers $k,D$, for any chordal graph $G$ of clique-tree degree at most $D$, can the connectivity of $TS_k(G)$ be decided in polynomial time?
\end{question}

\section{Split graphs}\label{sec:split}

This section is devoted to the proof of Theorem~\ref{th:hard}.
A graph $G=(V,E)$ is a \emph{split graph} if the vertices of $G$ can be partitioned into two sets $V_1,V_2$ such that the graph induced by $V_1$ is a clique and the graph induced by $V_2$ is an independent set. There is no restriction on the edges between $V_1$ and $V_2$. One can easily notice that split graphs do not contain any induced cycle of length at least $4$ since two non-adjacent vertices of such a cycle would belong to $V_1$, a contradiction.
Thus split graphs are chordal graphs.

\begin{figure}
 \centering
 \includegraphics[scale=0.6]{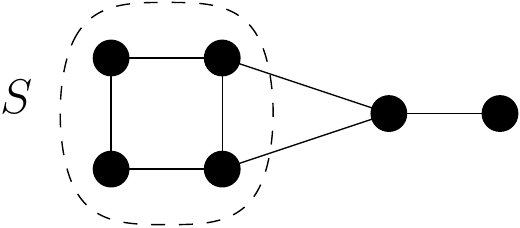}
 \caption{The set $S$ is blocking: no vertex in $N(S)$ is in the neighborhood of precisely one vertex of $S$. Note however that $S$ is not a dominating set.}
 \label{fig:blocking}
\end{figure}

Let $G=(V,E)$ be a graph. The \emph{neighborhood} of a vertex $v$, denoted by $N(v)$ is the set of vertices which are adjacent to $v$. The \emph{closed neighborhood of $v$}, denoted by $N[v]$ is the set $N(v) \cup \{ v \}$. A subset $S$ of vertices is a \emph{dominating set} of $G$ if $V = \cup_{s\in S} N[s]$. Deciding the existence of a dominating set of size $k$ (where $k$ is part of the input) is an NP-complete problem and a $W[2]$-problem when parameterized by $k$.

Let $S$ be a subset of vertices. A vertex $x$ is a \emph{private neighbor} of $s \in S$ if $x \in N(s)$ and $x \notin N[s']$ for every $s' \in S \setminus \{s\}$. We then say that $s$ has a private neighbor. A set $S$ is \emph{blocking} if no vertex of $S$ has a private neighbor with respect to $S$.
A graph $G$ is \emph{$k$-blocking} if it contains a blocking set of size at most $k$. In Figure~\ref{fig:blocking}, the set $S$ is blocking and then the graph is $4$-blocking.
Let us consider the following problem which is a variation of the domination problem:

\begin{quote}
\ds{} \\
\textbf{Input:} An integer $k$, a graph $G$ such that $G$ is not $k$-blocking. \\
\textbf{Output:} YES if and only if $G$ has a dominating set of size $k$.
\end{quote}

\begin{lemma}\label{lem:variantdom}
\ds{} is NP-complete.
\end{lemma}
\begin{proof}
 Checking whether a set is dominating can be performed in polynomial time. Thus \ds{} is in NP.
 In the remaining part of this proof, we argue that \ds{} is NP-complete. Let us show that there exists a polynomial time reduction from \dsc{} to \ds{}.
 
 Let $G=(V,E)$ be a graph with vertex set $\{v_1,\ldots,v_n\}$ and $k$ be an integer. If $k \geq n$ or $k \leq 3$, the problem can be decided in polynomial time. Thus we can assume that $n>k$ and that $k \geq 4$.

 We will construct a new graph $G'$ such that:
  \begin{enumerate}[(i)]
   \item $G'$ has no blocking set of size at most $k$.
   \item $G'$ has a dominating set of size at most $k$ if and only if $G$ has.
  \end{enumerate}
 
 The graph $G'$ is constructed as follows.
 
 \begin{itemize}
  \item Let $H$ be the disjoint union of $k$ copies $G_1,\ldots, G_k$ of $G$. We denote by $v_\ell^i$ the copy of $v_\ell$ in $G_i$. Let $H'$ be the graph obtained from $H$ by adding an edge between any two vertices that are copies of the same vertex or of two adjacent vertices in $G$ (i.e. we add the edges $(v_i^{j},v_i^\ell)$ and $(v_i^j,v_m^\ell)$ for any $i,j,\ell,m$ such that $(v_i,v_m)\in E$).
 \item Let $U$ be the disjoint union of $n$ induced paths on $(2k+2)$ vertices, and for every $i \in \{1,\ldots, n\}$, let $w_i$ be an arbitrary endpoint of the $i^{th}$ path. Let $U'$ be the disjoint union of $k(k-1)$ copies $U_{i,j}$ ($i, j \in \{1,\ldots,k\}$ with $i \neq j$) of $U$. We denote by $w_\ell^{i,j}$ the copy of $w_\ell$ in $U_{i,j}$.
 \item Let $G'$ be the graph obtained from the disjoint union of $H'$ and $U'$ by adding edges as follows. For every $i,j,p$, we add an edge between every vertex of $G_i$ and every vertex of $U_{i,j}$, and between $v_j^i$ and $w_j^{p,i}$. So $w_p^{i,j}$ is adjacent to all the vertices of $G_i$ and one vertex of $G_j$. The other vertices in $U_{i,j}$ are adjacent only to all the vertices of $G_i$.
 \end{itemize}
 Informally, point (i) and the sparsity between the sets $H'$ and $U'$  will ensure a dominating set of $G'$ is (essentially) contained in $H'$ and thus point (ii) is satisfied.
 For (i), the argument will be slightly more complicated. We essentially show that the lengths of the paths ensure that blocking sets must contain vertices in $G'$. The vertices of $G'$ are $k$ copies of the vertices of $G$. One can easily prove that any blocking set contains either zero or at least two vertices in each $G_i$. Thus a blocking set of size at most $k$ does not contain a vertex in each $G_i$, which leads to a contradiction.
 
Note that the graph $G'$ can be constructed in polynomial time.

 \begin{claim}\label{clm:noblocking}
  The graph $G'$ does not contain any blocking set of size at most $k$.
 \end{claim}
 
 \begin{proof}
 Assume by contradiction that $G'$ contains a set $S$ of size at most $k$ that is blocking.
 
 Assume first that $S$ contains a vertex $s$ of $U_{p,l}$ for some $p \neq l$ and $S$ does not contain any vertex of $G_p$.
 Let $P$ be a path of $U_{p,l}$ containing $s$. Since $|S| \leq k$ and the path $P$ has $(2k+2)$ vertices, there exist two consecutive vertices $\{a,b\}$ of $P$ which are not in $S$. Let $s'$ be the vertex of $S \cap P$ that is the closest from $\{ a,b \}$. Let $Q$ be a shortest path from $s'$ to $\{a,b\}$. Let us denote by $x$ the second vertex of $Q$. By construction, the vertex $x$ is in $N(s')$ and is not in the neighborhood of any vertex of $S \cap P$ distinct from $s'$. Indeed the other neighbor of $x$ in $P$ is either in $Q$ or in $\{a,b\}$; by definition no vertex of $Q \cup \{a,b\}$ is in $S$.
 Moreover, the vertex $x$ is not the beginning of $P$. Thus, no vertex of $S \setminus P$ is incident to $x$. Indeed, only vertices of $G_p$ are incident to vertices of $P$ that are not beginning of paths, and by assumption $S$ does not contain any vertex of $G_p$. Thus $s'$ has a private neighbor, a contradiction.
 
 Thus the set $S$ must contain at least one vertex $s$ in $\cup_{i=1}^k G_i$.
 Let $p$ be the index such that $s \in G_p$. Since $s$ is incident to all the vertices of $U_{p,l}$ and any vertex $x$ which is not in $G_p$ contains at most $3$ vertices of $U_{p,l}$ in its close neighborhood ($1$ if $x \in G_j$ with $j \neq p$, $3$ if $x \in U_{p,l}$ and $0$ otherwise), $s$ has a private neighbor if $|S \cap V(G_p)| = 1$. So we can assume that $S$ contains at least two vertices of $G_p$.
 Since $S$ has at most $k$ vertices and at least $2$ vertices of $S$ are in $G_p$, there exists an integer $i \leq k$ such that no vertex of $S$ is contained in  $G_i \cup_{l \neq i} U_{i,l}$.

 The vertex $s$ is incident to a vertex $u$ of $U_{i,p}$. Moreover, as we already mentionned, the vertex $u$ is a beginning of a path of $U_{i,p}$ and $x$ is the unique neighbor of $u$ which is not in $G_i \cup U_{i,p}$. Thus $u$ is a private neighbor of $s$, contradicting the definition of~$S$.
\end{proof}

\begin{claim}\label{clm:equivdom}
The graph $G$ has a dominating set of size $k$ if and only if $G'$ has a dominating set of size $k$.
\end{claim}
\begin{proof}
 Assume that $G$ has a dominating set $W=\{w_1,\ldots,w_k\}$ of size $k$. Let us prove that $W'=\{ w_1^1,\ldots,w_k^k \}$ is a dominating set of $G'$. Let $v_i^j$ be a vertex of $G_j$. Since $W$ is a dominating set, there exists a vertex $w_l$ of $W$ such that $(v_i,w_l)$ is an edge. By definition of $G'$, $(v_i^j,w_l^l)$ is an edge and then $v_i^j$ is dominated. Let $u$ be a vertex of $U_{p,l}$. By definition of $W'$, the vertex $w_p^p$ which is in $W'$ is completely connected to $U_{p,l}$ and then $u$ is dominated. Thus $W'$ is a dominating set of $G'$.
 
 Assume now that $G'$ has a dominating set $W'$ of size $k$. Then $W'$ contains exactly one vertex in each $G_p$. Indeed, assume by contradiction that $W'$ does not contain any vertex of $G_p$. Then every vertex of $W'$ contains at most $3$ vertices of $\cup_{l \neq p} U_{p,l}$ in its closed neighborhood (actually at most one if $W' \notin \cup_{l \neq p} U_{p,l}$ and $3$ otherwise). Since $U_{p,l}$ has size at least $|V| \cdot (2k+2) \geq k (2k+2)$, $S$ does not dominate $\cup_{l \neq p} U_{p,l}$. Thus $W'$ contains at least one vertex in each $G_p$ for $p\leq k$. Since $W'$ has size $k$, $W'$ contains exactly one vertex in each $G_p$.
 
 Let $v_{i_1}^1,v_{i_2}^2,\ldots,v_{i_k}^k$ be the vertices of $W'$. We claim that $W=\{ v_{i_1},v_{i_2},\ldots,v_{i_k} \}$ is a dominating set of $G$. Let $v$ be a vertex of $G$. Since $W'$ is a dominating of $G'$, there exists a $j$ such that $v_{i_j}^j$ is incident to $v^1$. By definition of $G'$, $(v_{i_j},v)$ is an edge, and then $v$ is dominated by $W$.
 \end{proof}
 The combination of Claims~\ref{clm:noblocking} and~\ref{clm:equivdom} ensure that Lemma~\ref{lem:variantdom} holds.
 \end{proof}

A $k$-independent set $I$ is \emph{frozen} in $TS_k(G)$ if $I$ is an isolated vertex of $TS_k(G)$. In other words, no token of $I$ can be slided. Note that $I$ is frozen if and only if $I$ is blocking. Indeed, if a vertex $u$ of $I$ has a private neighbor $v$ then the independent set $I \cup v\setminus u$ is incident to $I$ in $TS_k(G)$: the token of $u$ can be slided on $v$. Conversely, if $I$ and $J$ are incident in $TS_k(G)$ then a vertex of $I$ has a private neighbor. Indeed, assume that $J$ is incident to $I$ and let $u= I \setminus J$ and $v = J \setminus I$. Then $(u,v)$ is an edge. Since $J$ is an independent set, no vertex of $I \setminus u$ is incident to $v$, and then $v$ is a private neighbor of $u$.

\begin{theorem}\label{thm:conp}
 \tss{} is co-NP-hard.
\end{theorem}
\begin{proof}
 Let us prove that there is a reduction from \ds{} to \tss{} such that the first instance if positive if and only if the second one is negative. Since \ds{} is NP-complete by Lemma~\ref{lem:variantdom}, it implies that \tss{} is co-NP hard.
 
 \begin{figure}
  \centering
  \includegraphics[scale=0.75]{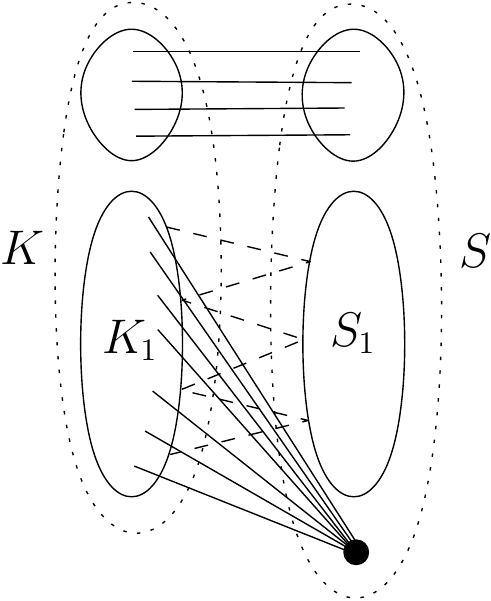}
  \caption{$K$ is the clique and $S$ is the independent set. The dashed part corresponds to the edges of the graph $G$. Then we add $k+1$ vertices on each side (left one inducing a clique) with a matching between the two sides. Finally the bottommost vertex of $S$ is the vertex $w_{n+k+2}$ which is connected to all the vertices of $K_1$.}
 \end{figure}

 Let $G$ be a graph on vertex set $\{v_1,\ldots,v_n \}$ that is not $k$-blocking. We construct a split graph $G'$ on vertex set $K \cup S$ where $K$ induces clique of size $n+k+1$ and $S$ induces an independent set of size $n+k+2$. The vertices of $K$ are denoted by $u_1,\ldots,u_{n+k+1}$ and the vertices of $S$ are denoted by $w_1,\ldots,w_{n+k+2}$. Moreover we will denote by $K_1$ (resp. $S_1$) the subset of $K$ (resp. $S$) $u_1,\ldots,u_{n}$ (resp. $w_1,\ldots,w_{n}$). The edges of the graph $G'$ are the following:
 \begin{itemize}
  \item For every $i,j \leq n$, $u_i$ is incident to $w_j$ if $i=j$ or $(v_i,v_j)$ is an edge of $G$. Thus the graph induced by $(K_1,S_1)$ simulates the incidence in the graph $G$.
  \item For every $n+k+2>j >n$, $w_j$ is the unique neighbor of $u_j$ in the independent set.
  \item The vertex $w_{n+k+2}$ is connected to the whole set $u_1,\ldots,u_n$.
 \end{itemize}
Let us prove that the $TS_{k+1}(G')$ is connected if and only if $G$ has no dominating set of size $k$.
First assume that $G$ has a dominating set $D$ of size $k$. Without loss of generality, we can assume that $D=\{v_1,\ldots,v_k\}$. The set $I=\{w_1,\ldots,w_k,w_{n+k+2}\}$ is an independent set of $G'$ of size $k+1$. Note that the independent set $I$ is frozen. Indeed, we have $N(I)= \{u_1,\ldots,u_n\}$. Moreover, since $D$ is a dominating set, for every vertex $u_i$, there exists a vertex of $\{ w_1,\ldots,w_k\}$ which is incident to $u_i$. Since $w_{n+k+2}$ is also incident to $u_i$, no vertex of $I$ has a private neighbor. And then no vertex of $I$ can be slided, \emph{i.e.} $I$ is the unique independent set of its own connected component in $TS_{k+1}(G')$. Since there are other independent sets of size at least $(k+1)$, $TS_{k+1}(G')$ is not connected. \smallskip

Assume now that there is no dominating set of size at most $k$. Let us prove that any independent set $J$ of size $k+1$ can be transformed into the independent set $I=\{w_{n+1},\ldots,w_{n+k+1}\}$. In order to prove it, let us argue that if $J \neq I$ then $J$ can be transformed into an independent set $J'$ such that $|J' \cap \{w_{n+1},\ldots,w_{n+k+1}\} | > |J \cap \{w_{n+1},\ldots,w_{n+k+1}\}$.

First assume that the vertex $w_{n+k+2}$ is in $J$.
Since the graph $G$ has no dominating set of size at most $k$, the set $J \setminus w_{n+k+2}$ does not dominate $\{u_1,\ldots,u_n\}$, thus $w_{n+k+2}$ has a private neighbor $u_i$. We can slide $w$ to $u_i$, then $u_i$ to some vertex $u_j$ for some $j>n$ such that $w_j$ is not in the independent set $J$, and finally move $u_j$ to $w_j$. At any step, we have an independent set and we have obtained the target independent set $J'$.

If $w_{n+k+2} \notin J$ and $J \cap \{w_1,\ldots,w_n\} \neq \emptyset$, we can slide a vertex of $J \cap \{w_1,\ldots,w_n\}$ to a vertex in $\{w_{n+1},\ldots,w_{n+k+1}\}$. Since $L=J \cap \{w_1,\ldots,w_n\}$ has size at most $k$, it is not blocking in $G$ thus neither in $G'$ by construction. Thus a vertex $w_i$ of $L$ has a private neighbor $u_j$ in $\{ u_1,\ldots,u_n\}$. So we can slide $w_i$ on $u_j$ and conclude as in the previous case.

Consider the final case, i.e. $J \cap \{w_1,\ldots,w_n\}=\emptyset$ and $w_{n+k+2} \not\in J$.  Since $J \neq I$, w.l.o.g. $J=\{u_1,w_{n+2},\ldots,w_{n+k+1}\}$. We can slide the token of $u_1$ to $u_{n+1}$ and then onto $w_{n+1}$, hence the conclusion.
\end{proof}

One can notice that, in the proof of Lemma~\ref{lem:variantdom}, the size of the dominating set does not change during the reduction. Moreover, in Theorem~\ref{thm:conp}, the difference between the size of the dominating set of $G$ in the original graph and the size of the independent sets we want to slide in the split graph is one. As a by-product, we immediately obtain the following corollary:

\begin{corollary}
 \tss{} is co-W[2]-hard.
\end{corollary}

\section{Interval graphs}

This section is devoted to the proof of Theorem~\ref{th:connec}.

\begin{figure}
 \centering
 \includegraphics[scale=0.65]{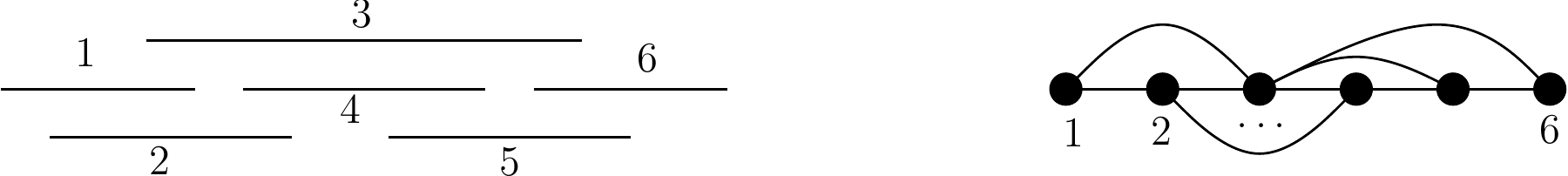}
 \caption{An interval graph on $6$ vertices with its representation.}
 \label{fig:interval}
\end{figure}

A graph $G$ is an \emph{interval graph} if $G$ can be represented as an intersection of segments on  the line (see Figure~\ref{fig:interval}). More formally, each vertex can be represented with a pair $(a,b)$ (where $a\leq b$) and vertices $u=(a,b)$ and $v=(c,d)$ are adjacent if the intervals $(a,b)$ and $(c,d)$ intersect.
Let $u=(a,b)$ be  a vertex; $a$ is the \emph{left extremity} of $u$ and $b$ the \emph{right extremity} of $u$. The left and right extremities of $u$ are denoted by respectively $l(u)$ and $r(u)$. Given an interval graph, a representation of this graph as the intersection of intervals in the plane can be found in $\mathcal{O}(|V|+|E|)$ time by ordering the maximal cliques of $G$ (see for instance~\cite{Booth76}). Actually, interval graphs admit clique paths and are thus a special case of chordal graphs.
Using small perturbations, we can moreover assume that all the intervals start and end at distinct points of the line. In the remaining of this section we assume that we are given a representation of the interval graph on the real line.

\subsection{Basic facts on independent sets in interval graphs}

\paragraph*{Leftmost independent set.}
We start with additional useful notions on interval graphs.
Let $G=(V,E)$ be an interval graph with its representation on the line. There are two natural orders on the vertices of an interval graph: the \emph{left order} denoted by $\prec_l$ and the \emph{right order} denoted by $\prec_r$. We have $u \prec_l v$ if and only if $l(u) \leq l(v)$ (note that by our small perturbations assumption, we never have $l(u)=l(v)$). Similarly, $u \prec_r v$ if $r(u) \leq r(v)$. Note that these two orders do not necessarily coincide.

Let $u$ be a vertex of $G$. The graph $G_u$ is the graph induced by all the vertices $v$ such that $l(v) > r(u)$. In other words, $G_u$ is the graph induced by the intervals located at the right of $u$ that do not intersect the interval $u$ (see Figure~\ref{fig:interval2} for an illustration). Similarly, $G^u$ is the graph induced by all the vertices $v$ such that $r(v) < l(u)$ (the graph induced by the intervals located at the left of $u$ that do not intersect the interval $u$).
The graph $G_u^w$ is the graph induced by all the vertices $v$ where $u \prec_l v \prec_l w$ and both $(u,v)$ and $(v,w)$ are not edges. In other words, $G_u^w$ is the graph induced by the intervals located between $u$ and $w$. We can alternatively define $G_u^w$ as the graph induced by the vertices both in $G_u$ and $G^w$.

\begin{figure}
    \centering
    \includegraphics[scale=0.65]{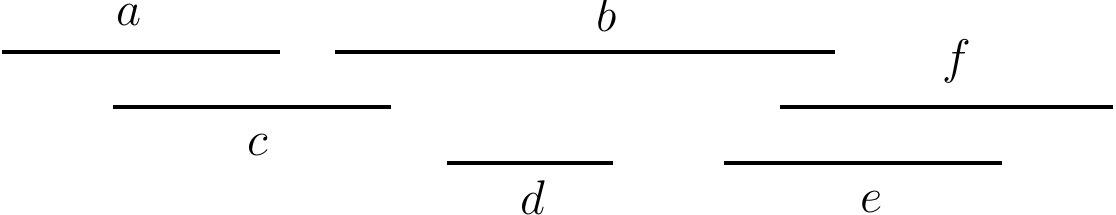}
    \caption{The leftmost independent set is $\{a,d,e\}$ while the rightmost one is $\{c,d,f\}$. The graph $G_c$ is restricted to the vertices $d,e,f$ and $G_c^e$ is restricted to the single vertex $d$.}
    \label{fig:interval2}
\end{figure}

The \emph{$\prec_r$-leftmost independent set} of $G$ is the independent set which can be recursively defined as follows:
\begin{itemize}
 \item If $G$ has no vertex, it is the empty set.
 \item Otherwise, let $u$ be the minimum vertex for the order $\prec_r$, \emph{i.e.}  the vertex $u$ that minimizes $r(u)$. The leftmost independent is $u$ plus the leftmost independent set of $G_u$.
\end{itemize}
One can similarly define the \emph{$\prec_l$-rightmost independent set} as the independent set which can be recursively defined as follows:
\begin{itemize}
 \item If $G$ has no vertex, it is the empty set.
 \item Otherwise, let $u$ be the maximum vertex for the order $\prec_l$, \emph{i.e.}  the vertex $u$ that maximizes $l(u)$. The rightmost independent is $u$ plus the rightmost independent set of $G^u$.
\end{itemize}

Note that if we reverse the order on the real line, the $\prec_r$-leftmost independent set becomes the $\prec_l$-rightmost independent. Thus all the following statements remain correct by replacing $\prec_r$-leftmost with $\prec_l$-rightmost independent sets. In order to avoid cumbersome notations and since we only consider $\prec_r$-leftmost independent sets and $\prec_l$-rightmost independent sets, we omit the $\prec_r$ and $\prec_l$ when we deal with them.

Given an independent set $u_1,\ldots,u_k$, if $l(u_i) < l(u_j)$ then $r(u_i) <l(u_j)$. So there is a natural order on the vertices of an independent set which corresponds to both $\prec_l$ and $\prec_r$. In the following, we say that the independent set is \emph{ordered} if $r(u_i) < l(u_{i+1})$ for every $i \leq k-1$. Moreover, by abuse of notation, when we  say that $l_1,\ldots,l_p$ is the leftmost independent set, we assume that the vertices are ordered.
Let us now state a couple of simple lemmas.

\begin{lemma}\label{lem:left}
Let $G$ be an interval graph and $l_1,\ldots,l_p$ be the leftmost independent set of $G$. Let $I=\{u_1,\ldots,u_k\}$ be an ordered independent set of size $k$. Then for every $j<k$, $r(u_j) \geq r(l_j)$.
\end{lemma}

\begin{proof}
 By contradiction. Let $j$ be the first index such that $r(u_j) < r(l_j)$. Since $r(u_{j-1}) \geq r(l_{j-1})$ (if it exists), the vertex $u_j$ is not incident to $l_{j-1}$ (since $l(u_j)>r(u_{j-1}) \geq r(l_{j-1})$). Thus $u_j$ is a vertex of $G[V \setminus N(\cup_{i=1}^{j-1} l_i)]$, contradicting the minimality of the right extremity of $l_j$.
\end{proof}

For every graph $G$ let us denote by $\alpha(G)$ the maximum size of an independent set of $G$. The value $\alpha(G)$ is called the \emph{independence number} of $G$.

\begin{lemma}\label{lem:leftmost}
 The leftmost independent set of $G$ has size $\alpha(G)$.
\end{lemma}

\begin{proof}
 By contradiction. Let $I=\{u_1,\ldots,u_{\alpha(G)} \}$ be an ordered independent set of size $\alpha(G)$. Let $J=\{ l_1,\ldots,l_p \}$ be the leftmost independent set of $G$ with $p<\alpha(G)$. Note that $u_{\alpha(G)}$ must be incident to $v_p$ otherwise $G[V \setminus N(\cup_{i=1}^{p} v_i)]$ is not empty. Thus $u_{\alpha(G)-1}$, and then $u_p$ satisfies $r(u_p) < r(v_p)$, a contradiction with Lemma~\ref{lem:left}.
\end{proof}

Lemma~\ref{lem:leftmost} ensures that (i) $TS_k(G)$ is empty if $k$ is larger than the size of the leftmost independent set (ii) $TS_k(G)$ is connected iff one can transform any independent set of size $k$ into the $k$ first vertices of the leftmost independent set. Our proof technique will precisely be based on point (ii).

\paragraph*{Token moves.}
We can then define the \emph{leftmost vertex} of an independent set $I$ as the minimum vertex of $I$ for both $\prec_l$ and $\prec_r$.

Before stating the next lemma, let us first make an observation on the representation of an independent set reconfiguration. We can represent a reconfiguration sequence as a sequence of independent sets $I_1,\ldots,I_m$ where $|I_i\setminus I_{i+1}|= |I_{i+1}\setminus I_i|=1$ and the unique vertex of $I_i\setminus I_{i+1}$ is incident to the unique vertex of $I_{i+1}\setminus I_i$. Let us denote by $u$ the unique vertex in $I_i\setminus I_{i+1}$ and by $v$ the unique vertex in $I_{i+1}\setminus I_i$. Thus the adjacency between these two independent sets can also be represented as the edge $(u,v)$. We say that a token is \emph{slided from $u$ to $v$} and that $(u,v)$ is \emph{the move} from $I_i$ to $I_{i+1}$. Thus a reconfiguration of independent set can be seen either as a sequence of adjacent independent sets in $TS_k(G)$ or as a sequence of moves.

Let $I_0$ be an independent set and $u \in I_0$. Let $I_0,\ldots,I_\ell$ be a reconfiguration sequence of independent sets. We define the \emph{token of origin $u$} in $I_k$ as follows: the token of origin $u$ in $I_0$ is $u$. For $t$ the token of origin $u$ in $I_{k-1}$, the token of origin $u$ in $I_k$ is either $t$ (if $t \in I_k$) or the vertex $v$ such that $(t,v)$ is the move from $I_{k-1}$ to $I_k$ (if $t \not\in I_k$).

On an interval graph, the \emph{leftmost token} of $I$ is the token on the leftmost vertex of $I$. More generally, there is a natural order on the tokens on vertices of $I$, as there is a natural order on the vertices themselves. The next lemma essentially ensures that the order is preserved by token sliding, i.e. we cannot ``permute'' tokens during a reconfiguration of an interval graph.

\begin{lemma}\label{lem:nonpermtoken}
Let $G$ be an interval graph and $I,J$ be two independent sets of $G$. Let $u$ be the leftmost vertex of $I$. For any reconfiguration from $I$ to $J$, the leftmost token of $J$ is the token of origin $u$.
\end{lemma}

\begin{proof}
Let $u$ be the leftmost vertex of $u$. Assume by contradiction that there exists a reconfiguration sequence from $I$ to $J$ where the token on the leftmost vertex of $J$ is not the token of origin $u$. Let $K$ be the first independent set of the sequence where the leftmost token is the token of origin $u$.
Let $L$ be the independent set before $K$ in the sequence. We have $L \setminus K = \{ u_1 \}$ and $K \setminus L = \{ u_2 \}$. Since the leftmost token of $L$ is not anymore the leftmost token of $K$, two cases may occur. Either $u_1$ is the leftmost vertex of $L$ and $u_2$ is not the leftmost vertex of $K$. Or $u_1$ is not the leftmost vertex of $L$ and $u_2$ is the leftmost vertex of $K$.

Let us first consider the first case. Consider the second leftmost vertex $v$ of $K$. Note that $v$ is in both $K$ and $L$ since only one token is slided at any step. Since $u_1$ is the leftmost vertex of  $L$, $r(u_1) < l(v)$. Since $v$ is the leftmost vertex of $K$, $r(v) < l(u_2)$. Moreover, since tokens are slided along edges, $(u_1,u_2)$ is an edge and then $r(u_1) > l(u_2)$. Thus $l(u_2) < r(u_1) < l(v) \leq r(v) < l(u_2)$, a contradiction.

Now assume that $u_1$ is not the leftmost vertex of $L$ and $u_2$ is the leftmost vertex of $K$. Let $v$ be the leftmost vertex of $L$. Note that $v$ is in both $K$ and $L$ since only one token is slided at any step. Since $v$ is the leftmost vertex of  $L$, $r(v) < l(u_1)$. Since $u_2$ is the leftmost vertex of $K$, $r(u_2) < l(v)$. Moreover, since tokens are slided along edges, $(u_1,u_2)$ is an edge and then $r(u_2) > l(u_1)$. Thus $l(u_2) < r(u_2) < l(v) \leq r(v) < l(u_2)$, a contradiction.
\end{proof}

The statement of Lemma~\ref{lem:nonpermtoken} can be easily generalized for any position token but we only need it for the first token in the remaining of the proof.
Using Lemma~\ref{lem:nonpermtoken}, we can prove the following statement.
 
 \begin{lemma}\label{lem:totheleft}
 Let $I$ and $J$ be two independent sets of $G$ with the same leftmost vertex $u$. If we can reconfigure $I$ into $J$ via a sequence of moves $\mathcal{P}$ that never slide the leftmost token (i.e. $u$ appears in all the independent sets of the sequence), then:
 \begin{itemize}
  \item $\mathcal{P}$ provides a reconfiguration from $I \setminus u$ to $J \setminus u$ in $G_u$.
  \item For every $v \prec_r u$, $\mathcal{P}$ provides a reconfiguration from $I \cup v \setminus u$ to $J \cup v \setminus u$ (in $G$).
 \end{itemize}
 \end{lemma}
 
 \begin{proof}
 Assume by contradiction that $\mathcal{P}$ does not provide a reconfiguration from $I \setminus u$ to $J \setminus u$ in $G_u$. It means that at some step, there exists a vertex $v$ in the current independent set that is not in $G_u$. Thus $l(v) < r(u)$. However, by Lemma~\ref{lem:nonpermtoken}, the leftmost token during the transformation from $I$ to $J$ in $G$ always is on $u$. Thus $l(v)$ must be at least $r(u)$, a contradiction.
 
 The proof of the second point is immediate since, at any step, the left extremity of the second token is larger than $r(u)$.  Since $r(v) <r(u)$, any move of $\mathcal{P}$ can be performed and thus the sequence of moves $\mathcal{P}$ transforms $I \cup v \setminus u$ into $J \cup v \setminus u$ (in $G$).
 \end{proof}

The proof of Lemma~\ref{lem:totheleft} is based on the fact that tokens cannot be permuted and that the leftmost token is never moved during the sequence. Thus we can perform exactly the same sequence of moves in $G_u$. Actually, we can generalize this argument to any sequence even if the leftmost token is moved.
Given a reconfiguration sequence $\mathcal{P}$ from $I$ to $J$, if we delete the leftmost vertex from each independent set of the sequence, we obtain a reconfiguration sequence from $I \setminus u$ to $J \setminus v$ where $u$ and $v$ respectively denotes the leftmost vertex of $I$ and $J$. Such a sequence can be extracted from $\mathcal{P}$ by omitting the moves in the sequence where the leftmost vertex is slided.
In Lemma~\ref{lem:totheleft}, we obtained the existence of the sequence in $G_u$. When the leftmost token is moved, we cannot hope for such a statement. However, we can prove that a sequence exists in $G_w$ for some well-chosen $w$.

\begin{lemma}\label{lem:totheleft2}
Let $I$ and $J$ be two independent sets such that there exists a transformation from $I$ to $J$. Let $u,v,w$ be respectively the leftmost token of $I$, the leftmost token of $J$ and the $\prec_r$-smallest leftmost token all along the sequence between $I$ and $J$.

There exists a transformation from $I \setminus  u$ into $J \setminus v$ in $G_w$.
\end{lemma}

\begin{proof}
Let $I$ and $J$ be two independent sets with leftmost vertices respectively $u$ and $v$. Let $\mathcal{Q}$ be a transformation from $I$ to $J$. Let $w$ be the $\prec_r$-smallest leftmost token during the sequence. Let us prove that there exists a transformation $\mathcal{Q}'$ of $I \setminus \{ u \}$ into $J \setminus \{ v\}$ in $G_w$.
Let us denote by $I_t$ the sequence after the $t$-th move of $\mathcal{Q}$ and $J_t$ independent set $I_t$ minus the leftmost vertex. Let us prove that the transformation $\mathcal{Q}$ can be transformed into $\mathcal{Q}'$ in such a way after $t$ steps of $\mathcal{Q}'$, the current independent set is $J_t$.

The construction of $\mathcal{Q}'$ is simple. Indeed, if the token slided between step $t$ and $t+1$ is the leftmost token, then we omit the transformation. Otherwise we perform the same transformation. Since $J_t$ is $I_t$ minus the first token, this transformation is possible in $G$ and leadts to $J_{t+1}$ which is $I_{t+1}$ minus the first token.

To conclude we just have to show that at any step, $J_t$ is included in $G_w$. Assume by contradiction that at some step $t$, $J_t$ is not included in $G_w$. Thus, the leftmost vertex $z$ of $J_t$ satisfies $l(z) < r(w)$. Let $y$ be the leftmost vertex of $I_t$. Since $I_t$ is an independent set, $r(y) < l(z)$. So $r(y) < r(w)$, contradicting the minimality of $w$.

Thus $\mathcal{Q}'$ provides a transformation of $I \setminus u$ into $J \setminus v$ in $G_w$.
\end{proof}

One can notice that by an immediate induction, a similar statement can be proved if we delete an arbitrary number of tokens (i.e. if we decide to delete the $\ell$-th first tokens). 

Using very similar arguments as in the proof of Lemma~\ref{lem:totheleft2}, one can prove the following statement.

\begin{lemma}\label{lem:totheright}
Let $I$ and $J$ be two independent sets such that there exists a transformation from $I$ to $J$. Let $w$ be the $\prec_l$-largest second token all along the sequence between $I$ and $J$.
Then the leftmost token of $I$ and the leftmost token of $J$ are in the same connected component of $G^w$.
\end{lemma}

\begin{proof}
Let $I$ and $J$ be two independent sets and $\mathcal{Q}$ be a transformation from $I$ to $J$.
Let  $w$ be the $\prec_l$-largest second leftmost vertex along the sequence. Let us prove that we can construct a transformation $\mathcal{Q}'$ turning the leftmost vertex of $I$ into the leftmost vertex of $J$ in $G^w$. Since these independent sets have size $1$, Lemma~\ref{clm:taille1} ensures that the two leftmost vertices are in the same connected component of $G^w$.

This transformation is derived from $\mathcal{Q}$ and at any step, the current unique vertex of the independent set is the leftmost vertex of the current independent set of $\mathcal{Q}$. More formally, let us denote by $I_i$ the sequence after the $i$-th move of $\mathcal{Q}$ and $u_i$ the leftmost vertex of $I_i$. Let $J_i$ be the sequence after the $i$-th move of $\mathcal{Q}'$. We want to show that $J_i$ is $\{ u_i \}$.

First note that $lm(I)$ and $lm(J)$ are in $G^w$. Thus the conclusion holds for $t=0$.
Let us now prove it by induction on $t$. If the move at step $t+1$ consists in moving a token that is not the leftmost one, we do not perform any move in $\mathcal{Q}'$ and then the current independent is still the leftmost vertex of $I_{t+1}$.

Otherwise, assume that $(u_t,u_{t+1})$ is the move. Let $w_t$ be the second leftmost vertex. By definition $l(w_t) < l(w)$. Since $u_{t+1}$ and $w_t$ are not incident ($I_{t+1}$ is an independent set) and tokens cannot be permuted by Lemma~\ref{lem:nonpermtoken}, $u_{t+1}$ is in $G^w$ and then the move $(u_t,u_{t+1})$ also exists, which concludes the proof.
\end{proof}

\subsection{Worst r-index}

Let $H$ be an interval graph and $I$ an independent set of $H$ of size $k$, we denote by $\mathcal{C}_{H,k}(I)$ the connected component of $I$ in $TS_k(H)$.
Given an independent set $J$, we denote by $lm(J)$ the leftmost vertex of $J$. Let

\[ RM(I,H) = \max_{\prec_l} \Big\{ lm(J) \ : \ J \in \mathcal{C}_{H,k}(I)\Big\}. \]
In other words, $RM_H(I)$ is the rightmost possible (for $\prec_l$) leftmost vertex of $J$ amongst all the independent sets $J$ in the component of $I$ in $TS_k(H)$. So if we try to push the leftmost vertex to the right, we cannot push it further than $RM(I,H)$.

Now, we can define the \emph{worst r-index} of $u$ for $H$ for independent sets of size $k$, denoted by $w_r(u,k,H)$, which is
\[ w_r(u,k,H) = \min_{\prec_l} \Big\{ RM(I,H) \ : \ I \text{ of size $k$ with leftmost vertex }u \Big\}. \]

When no independent set of size $k$ in $H$ has leftmost vertex $u$, we set for convenience $w_r(u,k,H)=+\infty$. Intuitively, $w_r(u,k,H)$ corresponds to the furthest to the right we can push a stable set in the worst case.

We define the symmetric notion consisting in pushing the leftmost vertex to the left (for $\prec_r$).
\[ LM(I,H) = \min_{\prec_r} \Big\{ lm(J) \ : \ J \in \mathcal{C}_{H,k}(I)\Big\}. \]

Now, we can define the \emph{worst l-index} of $u$ for $H$ for independent sets of size $k$, denoted by $w_l(u,k,H)$, which is
\[ w_l(u,k,H) = \max_{\prec_r} \Big\{ RM(I,H) \ : \ I \text{ of size $k$ with leftmost vertex }u \Big\}. \]
 Intuitively, $w_l(u,k,H)$ corresponds to the furthest to the left we can hope to push a stable set in the worst case.
 Note that when we want to push an independent to the left, we want to minimize $\prec_r$ while when we want to push an independent to the right we want to maximize $\prec_l$. \smallskip

\noindent\textbf{Example.}
\begin{figure}
    \centering
    \includegraphics[scale=0.65]{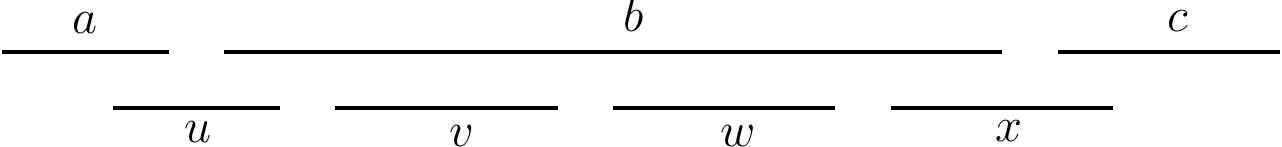}
    \caption{The worst r-index $w_r(v,2,H)$ is $w$. Indeed, if $v$ cannot move, then one cannot move $w$ on the right. On the other hand, the worst r-index $w_r(u,2,H)$ of $u$ is $c$ since any vertex can be slided to $z$ in $G_u$. In particular, worst r-indices are not necessarily increasing.}
    \label{fig:worstlefmost}
\end{figure}
Let us illustrate these notions on the graph $H$ given in Figure~\ref{fig:worstlefmost}. First note that we have $RM_H(\{v,w\})=v$. Indeed, no token in $\{v,w\}$ can be moved, thus the $\prec_l$-largest leftmost vertex of an independent set in the connected component of $\{v,w\}$ in $TS_k(H)$ is $v$ itself. In particular we have $w_r(v,2,H)=v$ since any independent set of size $2$ containing $v$ contains an independent set with leftmost vertex $v$.

Similarly, since no token in $\{v,w\}$ can be moved, the $\prec_r$-smallest leftmost vertex of an independent set  in the connected component of $\{v,w\}$ in $TS_k(H)$ is $v$ itself. Since $LM_H(\{v,w\})$ is indeed smaller or equal to $v$ (for $\prec_r$), we have $LM_H(\{v,w\})=v$. Since any independent set with leftmost vertex $v$ contains (indeed) $v$, we have $w_l(v,2,H) = v$.

Now consider any independent of size $2$ containing $u$, e.g. $\{u,v \}$. We can slide the vertex $u$ on $a$, and then the vertex $v$ on $b$, $x$ and finally $c$. Then we can slide the token on $a$ to $u$, $b$ and finally $w$. Thus $RM_H(\{u,v\})= w$. Since we can prove the same for any independent set of size $2$ containing $u$ and since $w$ is the second rightmost vertex of the rightmost independent set, we have $w_r(u,2,H)=w$. Note in particular that the worst r-index is not necessarily increasing. Indeed we have $u \prec_l v$ and $w_r(u) \succ_l w_r(v)$.

One can easily notice that any independent set of size $2$ containing $u$ as leftmost vertex can be slided on $a$. Thus $w_l(u,2,H)=a$ since $a$ is the minimum vertex for $\prec_r$.
\smallskip

We first argue how to compute $w_r(u,k,H_b)$ when $k=1$, for any vertices $u$ and $b$. To do it, we first need the following simple lemma.

 \begin{lemma}\label{clm:taille1}
Let $H$ be a graph.  The independent sets $\{u \}$ and $\{v \}$ are in the same connected component of $TS_1(H)$ if and only if $u$ and $v$ are in the same connected component of $H$.
 \end{lemma}
 
 \begin{proof}
Let $u,v$ be two vertices of $H$. Assume that $u$ and $v$ are in the same connected component. Thus there exists a path $u=u_1,\ldots,u_j=v$ between $u$ and $v$. The sequence $S_1=\{u \}, S_2=\{u_2\},\ldots, S_j=\{ v\}$ is a path of $TS_{1}(H)$ from $\{ u \}$ to $\{ v\}$ since at any step the token is slided along an edge. Similarly, any sequence from $\{ u \}$ to $\{ v\}$ can be transformed into a path by omitting the brackets in the sequence since the token is slided along an edge.
\end{proof}

From Lemma~\ref{clm:taille1}, we derive that for any two vertices $u,b$, if $u \in H_b$ then $w_r(u,1,H_b)$ is the rightmost vertex for $\prec_l$ in the connected component of $u$ in $H_b$. Since $w_r(u,1,H_b)=+\infty$ if $u \not\in H_b$, we can compute $w_r(u,1,H_b)$ in quadratic time (linear in the number of edges) for any two vertices $u,b$.

By symmetry, Lemma~\ref{clm:taille1} also ensures that for any two vertices $u,b$, if $u \in H_b$ then $w_l(u,1,H_b)$ is the $\prec_r$-leftmost vertex in the connected component of $u$ in $H_b$.

We now provide an algorithm that computes $w_r(u,k,H_b)$ for $k \geq 2$, assuming we can access in constant time all $w_r(v,\ell,H_c)$ for $\ell < k$ and any vertices $v$ and $c$. In order to define this algorithm, we need to extend the notion of $w_r$ to a set of vertices instead of a single verte: $w_r(S,\ell,H_c)=min_{\prec_l} \{w_r(j,\ell,H_c) | j \in S\}$.

\begin{lemma}\label{clm:procedure}
Let $H$ be an interval graph. The value $w_r$ returned by Procedure~\ref{alg:calculw} with input $u,k,H,b$ is $w_r(u,k,H_b)$.
\end{lemma}
 \begin{algorithm}[h!]
\caption{Computing $w_r(u,k,H_b)$ and $w_l(u,k,H_b)$, with $k \geq 2$ and assuming constant-time access to values for smaller sizes of the independent set}
\label{alg:calculw}
\begin{algorithmic}[PERF]
  \STATE $c:=u$.
  \STATE $r:=+\infty$.
  \WHILE{$r \neq c$}
  \STATE $r=c$.
  \STATE $j= w_r(V(H_u),k-1,H_c)$
  \STATE $c=$ the leftmost vertex (for $\prec_r$) in the connected component of $c$ in $H^j_b$.

  \ENDWHILE
  \RETURN $w_l=c$ and $w_r=$ the maximal vertex (for $\prec_l$) that can be reached from $c$ in $H_b^{j}$.
\end{algorithmic}
\end{algorithm}

Let us however briefly explain the behavior of Procedure~\ref{alg:calculw} and explain this procedure is correct.

Procedure~\ref{alg:calculw} essentially proceeds as follows: given an independent set $I$ containing $u$ as leftmost vertex, we try to ``push to the right'' the rest of the independent set without sliding $u$. In other words, we try to push an independent set of size $k-1$ to the right in $H_u$. Then we try to push $u$ to the left knowing that the second leftmost vertex may be in $w_r(u,k,H_a)$ (mentioned as $j$ in Procedure~\ref{alg:calculw}). Since $u$ has been pushed to the left (on the vertex $c$ when we follow the description of Procedure~\ref{alg:calculw}), it may make some space for the rest of the independent set. Then, instead of trying to push $I \setminus u$ to the right in $G_u$, we do it in $G_{c}$. We repeat this operation as long as the leftmost vertex moves to the left. This guarantees that there are at most $n$ iterations, each of them processed in quadratic time. Therefore, Procedure~\ref{alg:calculw} terminates in cubic time. When ``$c=u$'', it means that the rest of the independent set cannot be pushed further to the right, which in turn implies that the rightmost possible vertex we can reach from $u$ is also the rightmost possible vertex we can reach from $u$ in $H_a^{w_r(u,k,H_a)}$. Similarly, the leftmost possible vertex that can be reached to the left is the current vertex.

Keeping in mind this transformation, it is not hard to show that the value output by Procedure~\ref{alg:calculw} is smaller or equal to the worst r-index $w_r(u,k,H_a)$. We show that it is actually an equality by showing that the transformation of any independent set $I$ in $H_a$ with leftmost vertex $u$ can be described as above.

\begin{proof}
We call $(c_0,\ldots,c_t)$ the successive values taken by the variable $c$ in Procedure~\ref{alg:calculw} when we check if the while condition is satisfied. Note that $c_{t-1}=c_t$. Note moreover that the sequence $c_i$ is non increasing (for the order $\prec_r$).
Let us denote by $d$ the value output by the algorithm.

The proof is a proof by induction on $k$.

\begin{claim}\label{clm:superieurmodif}
For every independent set $I$ of size $k$ in $H_b$ with leftmost vertex $u$, the sets $I$ and $(I \setminus u) \cup c_t$ are in the same connected component of $TS_k(H_b)$.
\end{claim}
\begin{proof}
Let us prove by induction on $\ell \leq t$ that $I$ and $I_\ell:=(I \setminus u) \cup c_\ell$ are in the same connected component of $TS_k(H_b)$.
The proof is immediate for $\ell=0$ since $I_0=I$.
Assume now that $I_{\ell}$ and $I_0$ are in the same connected component of $TS_k(H_b)$.
Let us prove that $I_{\ell}$ and $I_{\ell+1}$ also are in the same connected component of $TS_k(H_b)$, which will conclude the proof.
Note that $I_\ell \setminus c_\ell$ is an independent set of size $k-1$ contained in the graph $H_{u}$.
By definition of $w_r(V(H_{u}),k-1,H_{c_\ell})$, any independent set $K$ of size $k-1$ with leftmost vertex $v$ satisfying $l(v) > r(u)$ can be slided to an independent set with leftmost vertex $w_r(V(H_{c_0}),k-1,H_{c_\ell})$ in $H_{c_\ell}$.
We can in particular do it for the independent set $I_{\ell} \setminus c_\ell$.
We then reach the independent set $J_\ell$.
By definition of $J_\ell$, it is possible to slide $c_\ell$ to  $c_{\ell+1}$ without sliding any other token since $c_\ell$ and $c_{\ell+1}$ are in the same connected component of $H_b^{lm(J_\ell)}$.
This is enough to connect them by Lemma~\ref{clm:taille1}. Since $c_{\ell+1} \prec_r c_\ell$, Lemma~\ref{lem:totheleft} ensures that the reverse of the sequence $\mathcal{P}$ transforms $J_\ell \cup c_{\ell+1}$ into $I_{\ell+1}$, which concludes the proof.
\end{proof}

\begin{claim}\label{clm:superieur}
 We have $w_r(u,k,H_b) \succ_r d$.
\end{claim}
\begin{proof}
By Claim~\ref{clm:superieurmodif}, any independent set of size $k$ with leftmost vertex $u$ can be transformed into the same independent set where the leftmost vertex is replaced by $c_t$. By definition of $w_r(V(H_{u}),k-1,H_{c_t})$, any independent set $K$ of size $k-1$ with a leftmost vertex $w$ satisfying $l(w) > r(u)$ can be slided to an independent set with leftmost vertex $w_r(V(H_{u}),k-1,H_{c_t})$. Thus there is an independent set with leftmost vertex $c_t$ and second leftmost vertex $w_r(V(H_{u}),k-1,H_{c_t})$. By definition of $d$ and by Lemma~\ref{clm:taille1}, $c_t$ can be slided to a vertex at least equal to $d$ (in the $\prec_l$ order) without any modification of the other vertices. Thus, for any independent set $I$ of size $k$ starting on $u$ in $H_b$, we have $RM(I,H_b) \succ_r d$. So finally we obtain $w_r(u,k,H_b) \succ_r d$.
\end{proof}

Let us now prove that there exists an independent set $I$ for which $RM(I,H_b)=d$.

\begin{claim}\label{clm:inferieur}
There exists an independent set $I$ of size $k$ with leftmost vertex $u$ for which $RM(I,H_b) \prec_r d$.
\end{claim}
\begin{proof}
Let $I_t$ be an independent set such that the leftmost vertex of $I_t$ is $c_t$ and all the other vertices of $I_t$ are in $H_{u}$ that minimizes $RM(I_t \setminus c_t, H_{c_t})$. Let us prove that $I= (I_t \setminus c_t) \cup u$ is an independent set of size $k$ with leftmost vertex $u$ such that $RM(I,H_b)=d$.

Let us prove the following stronger statement: any independent set $S$ in the connected component of $I_t$ in $TS_k(H_b)$ satisfies:
\begin{enumerate}[(i)]
 \item the leftmost vertex $u$ of $S$ satisfies $c_t \prec_r u$ and,
 \item the second leftmost vertex $v$ of $S$ satisfies $u \prec_l w_r(V(H_{u}),k-1,H_{c_t})$.
\end{enumerate}
Assume by contradiction that an independent set $S$ in the component of $I$ in $TS_k(H_b)$ does not satisfy (i) or (ii). Free to modify $S$, we can assume that every independent set in the sequence between $I$ and $S$ satisfies (i) and (ii).

Assume first that (i) does not hold. By minimality of $S$, the second leftmost vertex of $S$, denoted by $w$ satisfies $w \prec_l w_r(V(H_u),k-1,H_{c_t})$.
Let us denote by $S'$ the independent set before $S$ in the sequence. The vertex $w$ appears in both $S$ and $S'$ since only the leftmost vertex is slided between $S$ and $S'$ by Lemma~\ref{lem:nonpermtoken}. Thus $lm(S')$ can be slided from a vertex $y \succ_r c_t$ to a vertex $x \prec_r c_t$ that is not incident to $w$. In particular $y$ or $x$ must be incident to $c_t$ and both vertices are not incident to $w \prec_l w_r(V(H_u),k-1,H_{c_t})$. Thus the vertex $c_t$ can be slided to a smaller vertex (for $\prec_r$) than $c_t$ in $H_b^{w_r(V(H_u),k-1,H_{c_t})}$. Thus $c_{t+1}$ must be distinct from $c_t$, a contradiction with the fact that the algorithm stops at time $t$.

Assume now that (ii) is not satisfied in $S$. By minimality of $S$ in the sequence, the leftmost vertex of $S$ is larger than or equal to $c_t$. By Claim~\ref{clm:superieurmodif}, it is possible to slide the independent set $I$ into the independent set $I \setminus c_0 \cup c_t$. Let $\mathcal{P}$ be the sequence of moves that transform $I$ into $S$. By hypothesis on $S$, for any intermediate independent set $K$ we have $c_t \prec_r lm(K)$. Thus if we denote by $\mathcal{P}'$ the sequence of moves $\mathcal{P}$ where moves of the first token deleted, Lemma~\ref{lem:totheleft} ensures that the sequence $\mathcal{P}'$ also provides a transformation of $(I \setminus u) \cup c_t$ into an independent set where the second leftmost vertex is larger than $w_r(V(H_u),k-1,H_{c_t})$. Thus it is by Lemma~\ref{lem:totheleft} a transformation from $I \setminus u$ to $S$ without its leftmost vertex in $H_{c_0}$, a contradiction with the fact that $RM(I \setminus u,H_{c_t})= w_r(V(H_u),k-1,H_{c_t})$ is strictly smaller than the leftmost vertex of $S$ without its leftmost vertex (for $\prec_l$).
\end{proof}

The lemma is just a consequence of Claims~\ref{clm:superieur} and~\ref{clm:inferieur}.
\end{proof}

Knowing that Procedure~\ref{alg:calculw} outputs $w_r(u,k,G_{a})$, we can show that the second output value is $w_l(u,k,G_{a})$.
\begin{lemma}\label{lem:calculwl}
The value $w_l$ returned by Procedure~\ref{alg:calculw} is $w_l(u,k,H_{b})$.
\end{lemma}

\begin{proof}
Let $I$ be an independent set of $H_b$ with leftmost vertex $u$. Claim~\ref{clm:superieurmodif} ensures that at any step, the independent set $(I \cup c_t) \setminus u$ is in the connected component of $I$ in $TS_k(H_b)$. Thus the value $w_l$ output by the algorithm is not smaller than $w_l(u,k,H_b)$ (for the order $\prec_r$).

Let us now prove that  $w_l=w_l(u,k,H_b)$.
Let $I'$ be an independent set of size $k-1$ such that $RM(I',H_{w_l})=w_r(V(H_u),k-1,H_{w_l})$ and with leftmost vertex in $H_u$. By definition of $w_r(V(H_u),k,H_{w_l})$ such an independent set exists.
Let $I=I'\cup \{ u \}$.

Assume by contradiction that there exists an independent $J$ in the connected component of $I$ in $TS_k(H_a)$ such that the leftmost vertex $x$ of $J$ is $\prec_r$-smaller than $w_l$.
Since $RM(I \setminus \{ u \},H_{w_l})=w_r(V(H_u),k-1,H_{w_l})$, the second leftmost token $z$ satisfies at any step $l(z) \leq l(w_r(V(H_u),k-1,H_{w_l}))$ since at any step before $J$, the leftmost token is $\prec_r$-larger than $w_l$.

So Lemma~\ref{lem:totheright} ensures that if the leftmost token reaches the vertex $x$ at some point, then $x$ and $u$ are in the same connected component in $G^{w_r(v(H_u),k-1,H_{w_l})}_a$. Since $r(x) < r(w_l) \leq r(u)$, there also exists a path from $x$ to $w_l$.

But by definition of $w_l$, during the last loop of Procedure~\ref{alg:calculw}, the value of $c$ is not modified. At the beginning of the last loop the value of $c$ is $w_l$. And the value of the $c$ at the end of the loop is the $\prec_r$-leftmost vertex of $H^{w_r(v(H_u),k-1,H_{w_l})}_a$ that can be reached from $w_l$ and then from $u$.

Since $x$ can be reached and $r(x) < r(w_l)$, the value of $c$ is modified, contradicting the fact that the output is $w_l$.
\end{proof}

By combining Lemmas~\ref{clm:taille1},~\ref{clm:procedure} and~\ref{lem:calculwl}, we can compute in polynomial time all the values $w_r(u,k,H_b)$ and $w_l(u,k,H_b)$. Indeed, for any integer $k$, we only need to compute a cubic number of values (at most $k \cdot n^2$ to be precise). Moreover, each of these values can be computed in polynomial time for $k=1$ according to Lemma~\ref{clm:taille1} and for $k \geq 2$ according to Lemma~\ref{clm:procedure}.

To conclude the proof of Theorem~\ref{th:connec}, we just have to prove the following:

 \begin{lemma}\label{lem:polyTSk}
  Given all the values of $w_r(i,k,G_b)$ (for any two vertices $i,b$ and any integer $k$), we can determine in polynomial time whether $TS_k(G)$ is connected.
 \end{lemma}

The remaining of this section is devoted to proving Lemma~\ref{lem:polyTSk}. Let us prove that we can determine in polynomial time if all independent sets can be transformed into the leftmost independent set of corresponding size in polynomial time.

Let $G$ be an interval graph. Let $(\ell_1,\ldots,\ell_{\alpha(G)})$ be the leftmost independent set of $G$. Though $\ell_0$ is not defined, we set $G_{\ell_0}=G$.

\begin{lemma}
 The graph $TS_k(G)$ is connected if and only if, for every $i \leq k$ and for every $a \in G_{\ell_i}$, we have $w_l(a,k-i,G_{\ell_i})=\ell_{i+1}$.
\end{lemma}
\begin{proof}
First assume that for every $i \leq k$ and for every $a \in G_{\ell_i}$, we have $c(a,k-i,G_{\ell_i})=\ell_{i+1}$. Let $I=\{ u_1,\ldots,u_k\}$ be an independent set of size $k$. Let us prove that $I$ can be slided to the leftmost independent set $\{\ell_1,\ldots,\ell_k\}$.

In order to show it, let us prove by induction on $k$ that $I$ can be slided to $I_k= \{ \ell_1,\ldots,\ell_k,u_{i+1},\ldots,u_k \}$. Note that $I_0=I$ and then the conclusion holds for $i=0$. Assume now that the conclusion is true for some $i<k$ and let us prove that it holds for $i+1$.
Since $c(a,k-i,G_{\ell_i})=\ell_{i+1}$, there exists in the connected component of $I_i \setminus \{ \ell_1,\ldots,\ell_i\}$ in $TS_{k-i}(G_{\ell_i})$ an independent set with leftmost vertex $\ell_{i+1}$. Let us denote by $J_{i+1}$ this independent set.
The vertex $\ell_{i+1}$ is the leftmost possible vertex in $G_{\ell_i}$. Thus Lemma~\ref{lem:totheleft} ensures that the transformation of $I_i$ into $J_{i+1}$ can be transformed into a transformation of $I_i \cup \{ \ell_{i+1} \} \setminus \{ u_{i+1} \}$ into $J_{i+1} \setminus \ell_{i+1}$. The reverse of this transformation transforms $J_{i+1}$ into  $I_i \cup \{ \ell_{i+1} \} \setminus \{ u_{i+1} \}$. Thus by adding the vertices $\{\ell_1,\ldots,\ell_i\}$ in the independent set, we have a transformation from $I_i$ to $I_{i+1}$, and then the result holds.
\smallskip

Let us now prove that if  $TS_k(G)$ is connected then for every $i \leq k$ and for every $a \in G_{\ell_i}$, we have $c(a,k-i,G_{\ell_i})=\ell_{i+1}$. Let $J$ be an independent set whose of size $k-i$ with leftmost vertex $a$. Let $I$ be the independent set $J$ plus the $i$-th first leftmost vertices, i.e. $I=J \cup \{ \ell_1,\ldots,\ell_i \}$. Since $I$ can be slided into $\{ \ell_1,\ldots,\ell_k \}$ in $TS_k(G)$, Lemma~\ref{lem:totheleft} ensures that $J=I \setminus \{ \ell_1,\ldots,\ell_i \}$ can be slided into $\{ \ell_{i+1},\ldots,\ell_k \}$. Thus there exists, in the connected component of $J$ of $TS_{k-i-1}(G_{\ell_i})$ an independent set with leftmost vertex $\ell_{i+1}$. Since this holds for any independent set $J$ of $G_{\ell_i}$ of size at most $k-i-1$, we have $c(a,k-i,G_{\ell_i})=\ell_{i+1}$.
\end{proof}

Since all the values $w_l(a,k-i,G_{\ell_i})$ can be computed in polynomial time, Lemma~\ref{lem:polyTSk}, and then Theorem~\ref{th:connec} holds.

\subsection{Transformation between two independent sets}

Using similar techniques, we can decide in polynomial whether two independent sets $I$ and $J$ are in the same connected component of $TS_k(G)$ where $G$ is an interval graph.
We compute the values $LM(I,G_a)$ and $RM(I,G_a)$ for some fixed independent sets instead of computing the values $w_l$ and $w_r$. We prove that we only have to compute these values for a polynomial number of independent sets. We then check whether $LM(I,G)=LM(J,G)$ and if so, prove we can reduce the problem to deciding whether some sets $I',J'$ are in the same connected component of $TS_{k-1}(G')$ where $G'$ is a subgraph of $G$. 

Let us descrive the proof more formally. Let $I$ be an independent set. A subset $J$ of $I$ is a \emph{right subset} of $I$ if, when $J$ contains a vertex $x \in I$, then $J$ also contains all the vertices of $I$ larger than $x$ (for both $\prec_r$ and $\prec_l$). Let us prove that we can adapt Procedure~\ref{alg:calculw} in order to compute $RM(J,H_b)$ and $LM(J,H_b)$ for any right subset $J$ of $I$ included in $H_b$. This procedure, called Procedure~\ref{alg:calculw2}, is described below.

We first compute the values corresponding to the right subset corresponding to the single rightmost vertex of $I$. The values $RM(J,H_b)$ and $LM(J,H_b)$ can be computed in linear time (in the number of edges) if $J$ has size one by Lemma~\ref{clm:taille1}. We then apply Procedure~\ref{alg:calculw2} to compute the values for increasingly large right subsets.

\begin{lemma}\label{clm:procedure2}
Let $H$ be an interval graph and $I$ an independent set of $H$. Procedure~\ref{alg:calculw2} returns $LM(J,H_b)$ and $RM(J,H_b)$ with input $J,k,H,b$ where $J$ is a right subset of $I$.
\end{lemma}
 \begin{algorithm}[h!]
\caption{Computing $LM(J,H_b)$ and $RM(J,H_b)$ assuming constant-time access to values for proper right subsets of $J$.}
\label{alg:calculw2}
\begin{algorithmic}[PERF]
\STATE $u:=lm(J)$.
  \STATE $c:=u$.
  \STATE $r:=+\infty$.
  \WHILE{$r \neq c$}
  \STATE $r=c$.
  \STATE $j= RM(J \setminus u, H_c)$
  \STATE $c=$ the leftmost vertex (for $\prec_r$) in the connected component of $c$ in $H^j_b$.
  \ENDWHILE
  \RETURN $w_l=c$ and $w_r=$ the maximal vertex (for $\prec_l$) that can be reached from $c$ in $H_b^{j}$.
\end{algorithmic}
\end{algorithm}
\begin{proof}
The proof technique follows from the proof of Lemma~\ref{clm:procedure}. Let us denote by $c_t$ the value of $c$ after $t$ steps in the while loop. Let us first prove the following claim.

\begin{claim}\label{clm:superieurmodif2}
Both $J$ and $(J \setminus u) \cup c_t$ are in the same connected component of $TS_k(H_b)$.
\end{claim}
\begin{proof}
Let us prove it by induction on $t$. For $t=0$, we have $c_t=u$ and the conclusion holds.

Assume now that  $J$ and $(J \setminus \{ u \} ) \cup \{ c_t \}$ are in the same connected component of $TS_k(H_b)$. Let us prove it for $t+1$. By definition of $RM(J \setminus \{u \}, H_{c_t})$, one can move $J \setminus \{ u \}$ into an independent set of leftmost vertex $RM(J \setminus \{u \}, H_{c_t})$. This sequence of moves also transforms $(J \setminus \{ u \}) \cup \{ c_t \}$ into an independent set with leftmost vertex $u$ and second leftmost vertex $RM(J \setminus \{ u \}, H_{c_t})$. By Lemma~\ref{lem:totheright}, this latter independent set can be slided into an independent set with leftmost vertex $c_{t+1}$ without sliding any other token. Then there is a transformation from $J$ into an independent set $K$ with leftmost vertex $c_{t+1}$. Since $c_{t+1}$ is the $\prec_r$-smallest vertex on this sequence, Lemma~\ref{lem:totheleft2} ensures that we can transform $K$ into the independent set $(J \cup \{ c_{t+1} \}) \setminus \{ u \}$, which achieves the proof.
\end{proof}

By Claim~\ref{clm:superieurmodif2} there exist independent sets in the connected component of $J$ that have leftmost vertices $w_r$ and $w_l$. Indeed, we have shown that there is an independent set with leftmost vertex $w_l$. And since in its component there is an independent set with second leftmost vertex $j_{t+1}=  RM(J \setminus \{u \}, H_c)$, Lemma~\ref{lem:totheright} ensures that there exists an independent set with leftmost vertex $w_r$.

Let us now prove that the leftmost vertex $v$ of any independent set in the connected component of $J$ in $TS_k(G_a)$ satisfies $r(v) \geq w_l$ and $l(v) \leq  w_r$.

\begin{claim}\label{clm:tight2}
The leftmost vertex $v$ of any independent set in the connected component of $J$ in $TS_k(G_a)$ satisfies $r(v) \geq w_l$ and $l(v) \leq  w_r$.
\end{claim}
\begin{proof}
Assume by contradiction that there is an independent set $K$ in the component of $J$ with leftmost vertex $v$ satisfying $r(v) < r(w_l)$ or $l(v) >  l(w_r)$. Let $\mathcal{Q}$ be a transformation from $J$ to $K$. Free to extract a sub-sequence of $\mathcal{Q}$, we can moreover assume that $K$ is the first independent set of the sequence with this property.

First assume that $r(v)<r(w_l)$. Let $w$ be the second leftmost vertex of $K$. Since only one token is moved at any step, $w$ also appears in the independent set $K'$ before $K$ in the sequence (only the leftmost token is slided at that step). Since at any step of the sequence up to $K'$, the leftmost vertex is $\prec_r$-larger than $w_l$, we have a transformation of $J \setminus u$ into $K \setminus w$ in $G_{w_l}$ by Lemma~\ref{lem:totheleft2}. By definition of $RM(J \setminus u, H_{w_l})$, we have $l(w) \leq l(RM(J \setminus u, H_{w_l})$.
Let $v'$ be the leftmost vertex of $K'$. We have that $(v,v')$ is an edge, $r(v) < r(w_l)$ and $r(v') \geq r(w_l)$. Thus $w_l$ is adjacent to $v'$ or to $v$. Moreover $v'$ and $v$ are not adjacent to $w$ and consequently not to $RM(J \setminus u, H_{w_l})$ since  $l(w) \leq l(RM(J \setminus u, H_{w_l}))$.
It follows that there is a path from $w_l$ to $v$ in $G_a^{RM(J \setminus u, H_{w_l})}$, so $c_{t+1} \neq c_t$ at the final round of the algorithm, thus the output value is not $c_t=w_l$, a contradiction.

If $l(v) > w_r$, we can use similar arguments. Let $w$ be the second leftmost vertex of $K$. Since only one token is moved at any time, $w$ also appears in the independent set $K'$ before $K$ in the sequence. Since at any step of the sequence up to $K'$, the leftmost vertex is $\prec_r$-larger than $w_l$, we have a transformation of $J \setminus u$ into $K \setminus w$ in $G_{w_l}$ by Lemma~\ref{lem:totheleft2}. By definition of $RM(J \setminus u, H_{w_l})$, we have $l(w) \leq l(RM(J \setminus u, H_{w_l})$.
Let $v'$ be the leftmost vertex of $K'$. We have that $(v,v')$ is an edge, $l(v) > l(w_r)$ and $l(v') \leq l(w_r)$. Thus $w_r$ is either incident to $v'$ or to $v$. Moreover $v'$ and $v$ are not incident to $w$ thus not to $RM(J \setminus u, H_{w_l})$ since  $l(w) \leq l(RM(J \setminus u, H_{w_l}))$. There is a path from $w_r$ to $v$ in $G_a^{RM(J \setminus u, H_{w_l})}$, the algorithm cannot output $w_r$, a contradiction.

\end{proof}
The combination of both claims completes the proof of Lemma~\ref{clm:procedure2}.
\end{proof}

To conclude let us prove that we can determine in polynomial time if two independent sets are in the same connected component. Using Lemma~\ref{clm:taille1}, we can easily determine if two independent sets of size one are in the same connected component. Let us explain how we can recursively determine if two independent sets are in the same connected component of $TS_k(G)$.

\begin{lemma}\label{lem:samecomp}
Let $I$ and $J$ be two independent set of size $k \geq 2$. The independent sets $I$ and $J$ are in the same connected component of $TS_k(G)$ if and only if:
\begin{enumerate}
 \item $LM(I,G)$ and $LM(J,G)$ are the same, and
 \item Let $u$ the leftmost vertex of $I$, $v$ the leftmost vertex of $J$ and $a=LM(I,G)$. The independent sets $I \setminus \{ u \}$ and $J \setminus \{ v\}$ are in the same connected component of $TS_{k-1}(G_a)$.
\end{enumerate}
\end{lemma}
\begin{proof}
 The first point is indeed necessary. If $I$ has an independent set in its component with leftmost vertex $a$ and $J$ does not have such an independent set in its component, the components of $I$ and $J$ must be different. And conversely if they are in the same connected component, they must have the same worst l-index.
 
 Assume now that $I$ and $J$ satisfies $a=LM(I,G)=LM(J,G)$. First assume that $I$ and $J$ are in the same connected component. Let $\mathcal{Q}$ be a transformation from $I$ to $J$. Note that at any step of this transformation the right extremity of the leftmost vertex is at least $r(a)$ by definition of $a$. Lemma~\ref{lem:totheleft2} ensures that there is a transformation from $I \setminus \{ u \}$ to $J \setminus \{ v \}$. Thus the second point is necessary.
 
 Let us prove that it is sufficient. Assume that there is a transformation from $I \setminus \{ u \}$ into $J \setminus \{ v \}$ in $G_a$. Let $\mathcal{Q}$ be such a transformation. Let $\mathcal{P}_1$ be a transformation of $I$ into $(I \cup \{ a \} )\setminus \{ u \}$ and $\mathcal{P}_2$ a transformation of $J$ into $(J \cup \{ a \} )\setminus \{ v \}$. By Lemma~\ref{lem:totheleft2}, we can transform the independent set obtained after $\mathcal{P}_1$ into the independent set $I \cup \{ a \} \setminus \{ u\}$.
 We can similarly do the same for $J$. Since there is a sequence transforming $I \setminus \{ u \}$ into $J \setminus \{ v\}$ in $G_a$, this sequence transforms $I \cup \{ a \} \setminus \{ u\}$ into $J \cup \{ a \} \setminus \{ v\}$, which concludes the proof.
\end{proof}
Using recursively Lemma~\ref{lem:samecomp}, we can determine in polynomial time if two independent sets are in the same connected component, hence Theorem~\ref{th:pairsofinterval}.

\bibliographystyle{plain}

\bibliography{biblio}

\end{document}